\documentclass[conference]{IEEEtran}

\usepackage{algorithm}
\usepackage{algpseudocode}
\usepackage{amsmath}
\usepackage{amssymb}
\usepackage{amsthm}
\usepackage{dsfont}
\usepackage{caption}
\usepackage{comment}
\usepackage{diagbox}
\usepackage{etoolbox}
\usepackage{float}
\usepackage{import}
\usepackage{subcaption}
\usepackage{tabularx}
\usepackage{todonotes}
\usepackage{tikz}
\usetikzlibrary{arrows.meta, positioning,calc}

\newtheorem{theorem}{Theorem}%
\newtheorem{corollary}[theorem]{Corollary}%
\newtheorem{lemma}[theorem]{Lemma}%
\newtheorem{remark}{Remark}

\usepackage{pgfplots}
\pgfplotsset{compat=1.18}
\usepackage{pgfplotstable}
\theoremstyle{myremark}


\begin{document}

\title{Covert Multi-bit LLM Watermarking:\\ An Information Theory and Coding Approach}

\author{
\IEEEauthorblockN{
Sidong Guo\IEEEauthorrefmark{1},
Tyler Kann\IEEEauthorrefmark{1},
Teodora Baluta\IEEEauthorrefmark{2},
Matthieu R. Bloch\IEEEauthorrefmark{1}
}
\IEEEauthorblockA{\IEEEauthorrefmark{1}School of Electrical and Computer Engineering \quad
\IEEEauthorrefmark{2}School of Cybersecurity and Privacy\\
Georgia Institute of Technology, Atlanta, GA 30332\\
Email: sguo93@gatech.edu, kann@gatech.edu, teobaluta@gatech.edu, matthieu@gatech.edu
}
}

\date{\today}

\maketitle

\begin{abstract}
We study the problem of multi-bit watermarking for non-autoregressive large language models (LLMs). We introduce an information-theoretic model inspired by diffusion language models (DLM), in which the encoder has limited non-causal access to token distributions within each token block. This formulation enables an information-theoretic characterization of the non-causal watermarking capacity, in which knowledge of LLM cover statistics is leveraged to enable a multi-bit covert embedding. We study the information-theoretic limits of the model by combining Gelfand--Pinsker and channel synthesis coding techniques and obtain an exact characterization of the capacity. The embedding strategy is further optimized across blocks using a constrained Markov decision process (CMDP) and we develop an explicit algorithm based on polar codes following the information-theoretic principles. We simulate the error performance on LLaDA, and provide empirical total variation (TV) analysis as a function of key randomness.
\end{abstract}

\section{Introduction}\label{sec:introduction}
Textual watermarking has emerged as a promising defense against the misuse of large language models (LLMs), embedding a hidden signal by slightly perturbing the model's token distribution in a manner that remains statistically detectable to authorized parties~\cite{kirchenbauer2023watermark}. For every prompt \(\mathbf{x}_{\text{prompt}}\), LLM watermarking seeks to embed a message \(m\) with the aid of a secret key \(k\) by designing a biased token generation law \(x^n \sim \prod_{t=1}^n P(x_t \mid x^{t-1}, \mathbf{x}_{\text{prompt}}, m, k)\), while satisfying a distortion constraint relative to the base LLM distribution \(x^n \sim \prod_{t=1}^n P_\theta(x_t \mid x^{t-1}, \mathbf{x}_{\text{prompt}})\). The case \(m \in \{0,1\}\), also known as zero-rate watermarking, has been extensively studied~\cite{kirchenbauer2023watermark,christ2024undetectable,tsur2025heavywater}, often under additional robustness constraints~\cite{zhao2023provable}.

Building on this foundation, multi-bit watermarking aims to embed richer messages for applications such as source tracing and timestamping~\cite{lubin2003robust}. Existing approaches largely focus on algorithmic design, balancing trade-offs among robustness, distortion, and text quality. Representative techniques include semantic-invariance–based methods~\cite{yoo2023robust}, position allocation strategies~\cite{yoo2023advancing}, and optimization-based formulations under perplexity constraints~\cite{wang2023towards}. More recent work incorporates error-control coding and structured bit segmentation to enhance robustness~\cite{qu2025provably}. Despite these advances, the lack of a precise information-theoretic characterization of achievable embedding rates often results in suboptimal embedding rates of message into tokens.

To address this gap,~\cite{gilani2026arcmark} establishes a capacity result for distortion-free multi-bit watermarking by modeling the problem as a channel with causal encoder side information, capturing the fact that the watermark encoder has access to the base token distribution while the decoder does not. This perspective suggests a connection to the classical Gelfand--Pinsker (GP) framework studied in information-theoretic watermarking~\cite{moulin2003information}. However, its applicability to standard LLMs is limited by their strictly autoregressive structure, which precludes non-causal access to future token distributions. Recent LLM architectures based on diffusion and multi-token prediction~\cite{liu2024deepseek,gloeckle2024better,nie2025largelanguagediffusionmodels, you2026llada} have been shown to retain generation quality when token-wise autoregression is relaxed. This suggests that limited non-causal access to token distributions can serve as a useful abstraction bridging practical LLMs and information-theoretic models. Beyond enabling higher coding rates, such access also facilitates channel resolvability-based coding~\cite{han2002approximation,cuff2013distributed}, allowing the encoder to maintain covertness by ensuring that the induced distribution after embedding remains close to the base model.

The objective of the present work is to develop a covert LLM watermarking framework based on non-causal token distribution information. Our approach draws on prior work around channel resolvability analysis~\cite{han2002approximation,cuff2013distributed,song2016likelihood,bloch2013strong,bloch2016covert}, constrained Markov decision processes (CMDPs)~\cite{powell2007approximate,altman2021constrained,huang2025relatively}, and polar coding~\cite{korada2010polar,Chou2014d,chou2018empirical}. The main contributions are summarized as follows:
\begin{enumerate}
    \item We propose a multi-bit LLM watermarking model based on a non-autoregressive model for which, under suitable assumptions, the information-theoretic Gelfand--Pinsker and channel resolvability coding schemes can be adapted to ensure covertness of the watermarking.
    
    \item We develop an information-theoretic characterization of the the achievable message and key rates for the proposed model. A key technical component is a joint scheme realizing Gelfand--Pinsker coding and channel resolvability.
    
    \item We design a multi-bit watermarking algorithm by casting the problem as a CMDP as in~\cite{huang2025relatively} and constructing explicit resolvability-based polar coding schemes~\cite{arikan2009channel}.
\end{enumerate}

\textit{Notation:} Random variables and their realizations are denoted by uppercase and lowercase letters, respectively (e.g., \(X\) and \(x\)). Let \(X\) be a random variable taking values in a finite alphabet \(\mathcal{X}\). For \(n \in \mathbb{N}\), we write \(X^n \triangleq (X_1, \ldots, X_n) \in \mathcal{X}^n\), and for integers \(1 \le i \le j \le n\), we denote the subsequence by \(X_{i:j} \triangleq (X_i, \ldots, X_j)\). For a deterministic sequence \(\mathbf{x} \in \mathcal{X}^n\), its type (empirical pmf) is defined as \(P_{\mathbf{x}}(x) \triangleq \frac{1}{n} \sum_{t=1}^n \mathds{1}\{x_t = x\}\), and \(\mathcal{T}_P^n\) denotes the type class associated with a pmf \(P\)~\cite{thomas2006elements}. The Kullback--Leibler (KL) divergence between two probability mass functions \(P_1\) and \(P_2\) over \(\mathcal{X}\) is defined as \(\mathbb{D}(P_1 \,\Vert\, P_2) \triangleq \sum_{x \in \mathcal{X}} P_1(x) \log \frac{P_1(x)}{P_2(x)}\). Similarly, the total variation (TV) distance is \mbox{\(\mathbb{V}(P_1, P_2) \triangleq \frac{1}{2} \sum_{x \in \mathcal{X}} \big| P_1(x) - P_2(x) \big|\).}

\section{System Model}
We consider a non-autoregressive LLM with token vocabulary \(\mathcal{X} = [1:V]\), where the generated sequence is partitioned into \(B\) blocks of length \(n\). The joint distribution over the generated sequence factors across blocks as
\begin{equation}
P_{X^{nB}}(x^{nB}) = \prod_{b=1}^B P_{X^n}\big(x_{(b-1)n+1:bn}\big| x^{(b-1)n}, \mathbf{x}_{\text{prompt}}, \theta\big),
\end{equation}
where \(\theta\) denotes the underlying LLM parameters together with the base sampling rule. This abstraction is motivated by recent multi-token generation architectures, including self-speculative decoding and diffusion language models~\cite{liu2024deepseek, gloeckle2024better, nie2025largelanguagediffusionmodels, wang2025vocalnet}, which process multiple token positions jointly within each generation step. Most relevant to our formulation is that such architectures can provide the encoder with effective non-causal access to the token distributions within a block. For self-speculative decoding, this requires verification or reweighting of candidate tokens prior to sampling, whereas diffusion language models naturally expose distributions over multiple token positions through their non-causal denoising updates.

We first focus on embedding within a single block of length \(n\) as illustrated in Fig.~\ref{fig:system}. For each position \(t \in [1:n]\), define the state
\(
S_t = P_{X_t \mid \textnormal{context}} \in \mathcal{S} \triangleq \Delta^V,
\)
where \(\Delta^V\) denotes the probability simplex over \(\mathcal{X}\), and \(S_t\) represents the conditional token distribution induced by the underlying LLM. Let \(S^n \sim P_{S^n}\) denote the resulting state sequence. In the absence of message embedding, tokens are generated according to the distribution $P_{X^n|S^n}(x^n|s^n)$. Under independent assumption, this distribution simplifies to
\begin{equation}\label{base distribution}
    P_{X|S}^{\otimes n} = \prod_{t=1}^n P_{X|S}(x_t \mid s_t),
\end{equation}
where tokens are sampled independently conditioned on their corresponding states. 

Consider a random variable $U^n\in\mathcal{U}^n$ such that $P_{U^nX^nS^n} \triangleq W_{X^n|U^nS^n}P_{U^nS^n}$ has marginal $P_{X^n|S^n}P_{S^n}$. An encoder with non-causal access to \(S^n\), seeks to embed a message \mbox{\(M \in \mathcal{M} \triangleq \{0,1\}^\ell\)}, possibly using a secret key \mbox{\(K \in \mathcal{K} \triangleq \{0,1\}^r\)}. The encoding is performed as $f(S^n,M,K)$ for a publicly known function
\begin{align}
  f:(\Delta^V)^n\times\{0,1\}^\ell\times\{0,1\}^r\to \mathcal{U}^n
\end{align}
and generate tokens according to a biased distribution
\(W_{X^n|U^n,S^n}(x^n|f(s^n,m,k), s^n)\). \(\widetilde{P}^m_{X^n|S^n}(x^n | s^n) \) denotes the distribution induced by message $m$ conditioned on state $s^n$ averaged over keys as
\begin{align}
\frac{1}{|\mathcal{K}|} \sum_{k \in \mathcal{K}} 
    W_{X^n|U^n,S^n}\big(x^n | f(s^n,m,k), s^n\big).
\end{align}
The performance of the embedding is measured in terms of the average probability of error $\mathbb{P}\left(g(X^n,K)\neq M\right)$ and total variation (TV) $ \mathbb{V}\big( \widetilde{P}^m_{X^n S^n}, P_{X^nS^n} \big)$ where \(P_{X^nS^n} = P_{X^n|S^n} P_{S^n}\) and \(\widetilde{P}^m_{X^n S^n} = \widetilde{P}^m_{X^n|S^n} P_{S^n}\). A small TV imposes a joint covertness and semantic secrecy constraint against an observer without access to \(M\) or \(K\) \cite{bloch2016covert}.

A rate $R$ is achievable with key rate $R_K$ if there exists a sequence of functions $f$ and $g$ operating on increasingly large sequences of $n$ tokens such that
\begin{align}
 & \lim_{n\to\infty}\frac{\ell}{n} \geq R, \quad \lim_{n\to\infty}\frac{r}{n}\leq R_K\\ 
&\lim_{n\to\infty}\mathbb{P}\left(g(X^n,K)\neq M\right)=0   \\&\lim_{n\to\infty}\mathbb{V}\big( \widetilde{P}^m_{X^n S^n}, P_{X^nS^n} \big)=0 \,\forall m.
\end{align}
Unlike distortion-free formulations~\cite{gilani2026arcmark, wang2008perfectly} that require exact matching ($\mathbb{V}\big( \widetilde{P}^m_{X^n S^n}, P_{X^nS^n} \big)=0$), the total variation  constraint allows an asymptotically vanishing but non-zero statistical deviation. 
The supremum of achievable rates is called the covert multi-bit LLM watermarking capacity, for non-autoregressive models.
\label{sec:intro}
\begin{figure}
  \centering
  \begin{tikzpicture}[
    box/.style={draw, rectangle, minimum width=1.2cm, minimum height=0.5cm, thick},
    arr/.style={-{Latex[length=2mm]}, thick},
    node distance=1.0cm
    ]
    
    \node[box] (f) {$f(\cdot)$};
    \node[box, below=1.0cm of f, xshift=-0.0cm] (W) {$W_{X^n|U^n,S^n}$};
    \node[box, right=2.0cm of W] (g) {$g(\cdot)$};
    
    \node[above=0.5cm of f] (M) {$M$};
    
    \node[right=1.4cm of f] (K) {$K$};
    
    \node[left=1.0cm of W, yshift=0.0cm] (S) {$S^n$};
    
    \node[right=0.5cm of g] (Mhat) {$\hat{M}$};
    
    \draw[arr] (M) -- (f);
    \draw[arr] (K) -- (f);
    
    \coordinate (Sjunc) at ($(S) + (0.7, 0)$);
    \fill (Sjunc) circle (1.5pt);  
    \draw[arr] (S) -- (Sjunc) -- (W.west);
    \draw[arr] (Sjunc) -- (Sjunc |- f.west) -- (f.west);
    
    \draw[arr] (f.south) -- node[right] {} (W.north);
    
    \coordinate (Kjunc) at (g.north |- K.east);
    \draw (K) -- (Kjunc);
    \draw[arr] (Kjunc) -- (g.north);
    
    \draw[arr] (W) -- node[above] {$X^n$} (g);
    
    \draw[arr] (g) -- (Mhat);
    
  \end{tikzpicture}
  \caption{Information-theoretic model of multi-bit LLM watermarking with multi-token prediction.}
  \label{fig:system}
\end{figure}
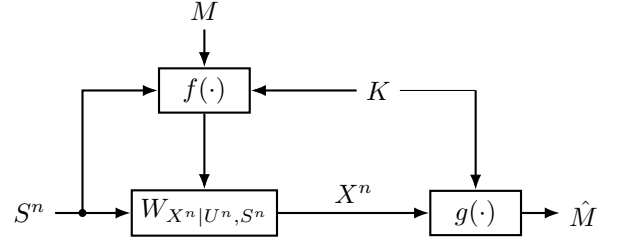

\section{Main Results}
We first present a covert capacity result for the case in which the states are independently and identically distributed, \mbox{\(S^n \sim P_{S}^{\otimes n}\)} with independent token generation rule \(P_{X|S}^{\otimes n}\). 

\begin{theorem} \label{MainTheorem}
  The covert multi-bit LLM watermarking capacity with non-causal state is
  \begin{align}
    \max_{P(u|s),W(x|u,s)} &\mathbb{I}(U;X)-\mathbb{I}(U;S)\label{eq:1}\nonumber\\
    &\mathrm{s.t.} \sum_u P_{UXS}(u,x,s)=P_{XS}(x,s)
  \end{align}
and is achievable with key rate $R_K=\mathbb{I}(U;X|S)$, where the random variables are distributed according to the maximizing distribution in~(\ref{eq:1}).
\end{theorem}

\begin{proof}
  See~\cite{arXivLLM}. The achievability proof relies on a likelihood encoder and soft-covering arguments~\cite{han2002approximation,song2016likelihood}. The converse is established by combining standard GP chaining techniques with TV-based resolvability analysis.
\end{proof}

Note that the capacity result in Theorem. \ref{MainTheorem} can be extended to the general memory setting with arbitrary joint distribution \(P_{U^nX^nS^n}
\triangleq
W_{X^n|U^nS^n}
P_{U^n|S^n}
P_{S^n},\)
using information-spectrum methods~\cite{koga2013information}, where the mutual information terms are replaced by spectral-inf and spectral-sup mutual information rates~\cite{bloch2013strong}. 

To build intuition on the operational meaning, we next consider a simple example inspired by~\cite[Corollary~3.2]{gilani2026arcmark}, where we explicitly derive and compute the achievable rate for a binary token distribution under the setting of Theorem~\ref{MainTheorem}.

\begin{corollary} \label{corollary 5.2}
\textit{Let} \(X \in [1\!:\!V]\) \textit{and define the set of token pairs} 
\(\{\{i,j\}: 1 \leq i < j \leq V\}\). \textit{Consider states for which the base distributions satisfy}
\begin{equation}\label{states}
    P_{X|S:\,\mathrm{supp}(S)=\{i,j\}}(x) = 
    \begin{cases}
        \frac{1}{2}, & x \in \{i,j\},\\
        0, & \text{otherwise}.
    \end{cases}
\end{equation}
\textit{Assuming uniform distribution over states in (\ref{states}), choosing} \(|\mathcal{U}|=2\) \textit{an achievable rate with binary token distributions is}
\begin{equation}
    R = \frac{\lfloor V^2/4 \rfloor}{\binom{V}{2}}, 
    \qquad \lim_{V \to \infty} R = 0.5~\text{bits/token},
\end{equation}
\textit{where the key rate satisfies}
\begin{equation}
    R_K = R = \frac{\lfloor V^2/4 \rfloor}{\binom{V}{2}},
\end{equation}
\textit{which is sufficient to ensure joint semantic secrecy and covertness.}
\end{corollary}

\begin{proof}
  See~\cite{arXivLLM}
\end{proof}

Corollary~\ref{corollary 5.2} highlights the theoretical benefit of non-causal encoding under an i.i.d. state distribution, compared to causal schemes~\cite{gilani2026arcmark}. The optimal strategy can be interpreted as partitioning the vocabulary into two balanced subsets and encoding information through the subset index, a pattern that also extends to state with memory.

\begin{remark}
\textit{In multi-block settings, the decoder must align with the encoder’s learned joint distribution \(P_{S^nU^nX^n}\) for each block, which we assume is public information. This requirement does not grant full access to prompt or LLM model, but instead assumes knowledge of the induced support on \(X^n\), the support and associated conditional probabilities for each state sequence \(S^n\). While this deviates from standard assumptions in LLM-based watermarking, it represents an intermediate regime between full model access and purely black-box decoding. This setting is relevant in deployment scenarios such as synthetic data generation~\cite{sander2024radioactive}, benchmark watermarking~\cite{sander2025benchmark}, and LLM steganography~\cite{kaptchuk2021meteor, dewitt2023steganography}. In the following sections, we explicitly operate under this regime, and the proposed algorithmic approach is designed to approximate and leverage such structured distributions.}
\end{remark}

\section{Implementation Approach}
Building upon the ideas from \cite{huang2025relatively}, we first explain how to evaluate and optimize the joint distribution \(P_{S^nU^nX^n}\) over multiple blocks using CMDP, then construct a polar code scheme that operates over short blocklength based on the learned joint distribution \(P_{S^nU^nX^n}\). 
\subsection{CMDP Formulation}
Theorem~\ref{MainTheorem} characterizes the covert capacity of LLM watermarking under asymptotic non-causal access to the cover distribution at the encoder. In practice, however, computational constraints and the local-context dependence of LLMs limit multi-token prediction to finite block lengths. Consequently, GP coding must operate over short blocks, and the message embedding scheme must select token distributions that balance immediate performance against impact on future state evolution~\cite{huang2025relatively}. We follow the approach of \cite{huang2025relatively} and formulate the problem as a CMDP. Unlike \cite{huang2025relatively} which targets the entropy of the induced distribution when decoder has access to token distribution, our objective is a GP rate that depends on the action and policy-induced state distribution. This is in line with the observation that explicit capacity-achieving constructions, such as polar codes, require knowledge of the optimized joint law \(P_{SUX}\)~\cite{thomas2006elements}. 

Assume that token generation is divided into $B$ blocks of length $L$. Let $X^{(b)} \in \mathcal{X}^L$, $U^{(b)} \in \mathcal{U}^L$, and $S^{(b)}$ denote the token, auxiliary, and token distribution sequences at block $b$, where the token distribution update follows LLM parameter \(\theta\) as \(S^{(b)}= f_\theta(X^{(b-1)},S^{(b-1)})\). The relevant CMDP components are specified in Table~\ref{tab:CMDP}.
The augmented state $Z^{(b)}$ is the distribution of \(S^{(b)}\) induced by the prompt prior $P_{\text{prompt}}$ and the policy-induced randomness in previous blocks $X^{(1:b-1)}$. The CMDP state $Z^{(b)}=\mathcal{T}^{\pi_\phi}_{\theta}(Z^{(b-1)})$ is a function of policy $\pi_\phi$ and LLM parameters $\theta$. Given $P_{\text{prompt}}$ and the resulting $Z^{(1)}$, we seek the optimal $P^\phi_{U^{(b)}|S^{(b)}}$ and $W^\phi_{X^{(b)}|U^{(b)},S^{(b)}}$ for $b \in [1{:}B]$ induced by the policy \(\pi_\phi\). As a result, the rate-covertness tradeoff is given by
\begin{equation}\label{optimization problem}
\begin{aligned}
\max_{\{P^\phi_{U^{(b)}|S^{(b)}},\, W^\phi_{X^{(b)}|U^{(b)},S^{(b)}}\}_{b=1}^B} \quad
& \mathbb{E}_{\pi_\phi}\Bigg[\sum_{b=1}^B\gamma^b 
\Big( r(z^{(b)},a^{(b)}) \Big)\Bigg] \\
\text{s.t.} \quad
 \mathbb{E}_{\pi_\phi} &\bigg[\sum_{b=1}^B\gamma^b \big(c(z^{(b)},a^{(b)})\big)\bigg] \leq \epsilon.
\end{aligned}
\end{equation}

The CMDP formulation provides an algorithmic approximation to the memory-based multi-letter extension of the capacity expression in Theorem~\ref{MainTheorem} \cite{koga2013information}, where reward and cost capture GP rate and TV. Algorithm~\ref{alg:pd_mb} enables a practical search over encoding distributions under finite-block and computational constraints \cite{tessler2018reward}, building on similar techniques used in~\cite{nikbakht2024memory}.

The policy $\pi_\phi$ maps the current state $Z^{(b)}$ to the action 
$a^{(b)}$, 
where the conditional distributions are parameterized by softmax heads on a learned state--block embedding. The action lies in the product-of-simplices 
$\prod_{s\in \textnormal{supp}(Z^{(b)})}\big(\Delta^{|\mathcal{U}|^L} \times (\Delta^{|\mathcal{X}|^L})^{|\mathcal{U}|^L}\big)$. In practice we retain a finite candidate set to reduce the dimensionality.
In our simulations, we demonstrate at short token lengths, for which the state distribution update can be carried out exactly
\begin{equation}\label{augmentationstate}
\begin{aligned}
    Z^{(b+1)}(s') = \sum_{s} Z^{(b)}(s) \sum_{u} P^\phi_{U^{(b)}|S^{(b)}}(u \mid s)\, \\ \sum _{x\in\textnormal{supp}(s):f_\theta(x, s)=s'}W^\phi_{X^{(b)}|U^{(b)},S^{(b)}}(x \mid u, s).
\end{aligned}
\end{equation}

\setlength{\abovecaptionskip}{4pt}
\setlength{\belowcaptionskip}{4pt}
\begin{table}[t]
\centering
\caption{CMDP Formulation}
\begin{tabular}{|c|c|}
\hline
State & $Z^{(b)}=P_{S^{(b)}}$ \\
\hline
Action & $a^{(b)}\triangleq\pi_\phi(Z^{(b)})=(P^\phi_{U^{(b)}|S^{(b)}},\,W^\phi_{X^{(b)}|U^{(b)},S^{(b)}})$ \\
\hline
Reward & $r^{(b)}=\mathbb{I}(U^{(b)};X^{(b)})-\mathbb{I}(U^{(b)};S^{(b)})$ \\
\hline
Cost & $c^{(b)}=\mathbb{V}( \widetilde{P}_{X^{(b)}S^{(b)}}, P_{X^{(b)}S^{(b)}})$ \\
\hline
\end{tabular}
\label{tab:CMDP}
\end{table}

\begin{algorithm}[t]
\caption{Primal--Dual Policy Gradient}
\label{alg:pd_mb}
\begin{algorithmic}[1]
\Require $\gamma, \epsilon, \eta_\phi, \eta_\beta$, policy $\pi_\phi$, known transition law $\mathcal T^{\pi_\phi}_{\theta}$
\State Initialize $\phi$ and $\beta \ge 0$
\For{iteration $t = 0,1,\dots,T-1$}
    \State Initialize $Z^{(1)}$
    \For{$b = 1,\dots,B$}
        \State $a^{(b)} \gets \pi_\phi(Z^{(b)})$
        \State $r^{(b)} \gets \mathbb{I}(U^{(b)};X^{(b)})-\mathbb{I}(U^{(b)};S^{(b)})$
        \State $c^{(b)} \gets \mathbb{V}\!\big(\widetilde P_{X^{(b)}S^{(b)}},P_{X^{(b)}S^{(b)}}\big)$
        \State $Z^{(b+1)} \gets \mathcal T^{\pi_\phi}_{\theta}(Z^{(b)})$ as in (\ref{augmentationstate})
    \EndFor
    \State $\phi \gets \phi+\eta_\phi \nabla_\phi \sum_{b=1}^B \gamma^b\big(r^{(b)}-\beta c^{(b)}\big)$
    \State $\beta \gets \big[\beta+\eta_\beta\big(\sum_{b=1}^B\gamma^b c^{(b)}-\epsilon\big)\big]_+$
\EndFor
\end{algorithmic}
\end{algorithm}

\subsection{Polar Code Construction}

Algorithm. \ref{alg:pd_mb} provides the encoder with the joint distribution for each block \(P_{S^{(b)} U^{(b)} X^{(b)}}= P_{S^{(b)}} P_{U^{(b)}|S^{(b)}} P_{X^{(b)}|U^{(b)},S^{(b)}}, \forall b\in [1:B]\), which we use to construct a polar coding scheme that exploits the non-causal state information. 

A polar code~\cite{arikan2009channel} is created through the base polarization matrix $G$ = $\left[ {\begin{array}{cc}
    1 & 0 \\
    1 & 1 \\
  \end{array} } \right],$ where $G^{\otimes p}$ is the $p$-th order Kronecker product of $G$ and $L = 2^p$ is the resulting blocklength. Specifically, given a vector $V$, a codeword $U$ is created through the operation $U =  VG^{\otimes p}$, a process known as the polar transform. The structure of the matrix $G^{\otimes p}$ allows for an encoding complexity of $O(L\log L)$. The polar transform converts $L$ physical channels into $L$ synthetic channels, where, asymptotically, the synthetic channels are either extremely high or low entropy when decoded successively. For some decoding threshold \(T_\delta, T_\epsilon >0\), using the polar transform~\cite{arikan2009channel}, the encoder constructs index sets over the pre-transform vector \(V\) as
\begin{equation} \label{polar set}
\begin{aligned}
    &\mathcal{H}_{\mathrm{enc}} = \{i : \mathbb{H}(V_i \mid V^{i-1}, S^L) > 1 - T_\epsilon\}, \\
    &\mathcal{L}_{\mathrm{enc}} = \{i : \mathbb{H}(V_i \mid V^{i-1}, S^L) \leq 1 - T_\epsilon\}, \\
    &\mathcal{H}_{\mathrm{dec}} = \{i : \mathbb{H}(V_i \mid V^{i-1}, X^L) \geq T_\delta\}, \\
    &\mathcal{L}_{\mathrm{dec}} = \{i : \mathbb{H}(V_i \mid V^{i-1}, X^L) < T_\delta\},
\end{aligned}
\end{equation}
where these sets are chosen following the framework developed in \cite{Chou2014d}. Specifically, the set \(\mathcal{H}_{\mathrm{enc}}\) contains indices with high conditional entropy with respect to state \(S^L\) at the encoder, where bits can be freely assigned without significantly altering the induced distribution of \(U^L\). The set \(\mathcal{L}_{\mathrm{dec}}\) consists of indices with low conditional entropy with respect to \(X^L\) at the decoder, which can therefore be reliably decoded. Accordingly, the message index set is defined as
\(
\mathcal{M}^* = \mathcal{H}_{\mathrm{enc}} \cap \mathcal{L}_{\mathrm{dec}}.
\)

To generate \(V\), the encoder can freely choose the bits corresponding to indices in \(\mathcal{H}_{\textrm{enc}}\), while the bits in \(\mathcal{L}_{\textrm{enc}}\) must be sampled according to the successive cancellation (SC) encoding rule. In the standard setting, this rule takes the form
\begin{align}
 \mathbb{P}(V_i) = \mathbb{P}(V_i \mid S^L, V_{0:i-1})\quad\forall i \in \mathcal{L}_{\textrm{enc}} .\label{SCrule}
\end{align}
However, the rule does not directly apply in our scenario because of the presence of channel memory, which prevents factorization into per-bit log-likelihood ratios. Specifically, the encoder is constrained by probabilities over entire sequences \(U\), rather than i.i.d bit-wise distributions \(U_i\). In DLMs, for instance, the token marginals are recomputed after
every committed token, so the block distribution does not factorize.

To account for this non-i.i.d dependence across the sequence, we write the SC rule in \eqref{SCrule} as
\begin{align}
\mathbb{P}(V_i) = \sum_{V_{j},\, j \in \mathcal{L}_{\textrm{enc}},\, j>i}
 \mathbb{P}\big(V_{i:L-1} \mid S^L, V_{0:i-1}\big)
 \quad \forall i \in \mathcal{L}_{\textrm{enc}}. \label{SCrule_modified}
\end{align}
Let $m \in \{0,1\}^{|\mathcal{M}^*|}$ denote the block of message bits to be embedded,
placed at the indices $\mathcal{M}^*$ of $V$, and write $V_{\mathcal{M}^*}$ for the
sub-vector of $V$ indexed by $\mathcal{M}^*$.
Since the polar transform is a bijection and \(\mathbb{P}(U = u) = \mathbb{P}(V = uG^{\otimes p})\), (\ref{SCrule_modified}) is equivalent to sampling from the set 
\(
\mathcal{V}^* = \{ V \in \mathbb{F}_2^L : \mathcal{V}_{\mathcal{M}^*} = m \}
\)
according to the marginalized distribution
\[
\frac{\mathbb{P}(V|S^L)\mathds{1}\{V \in \mathcal{V}^*\}}{\sum_{V' \in \mathcal{V}^*} \mathbb{P}(V'|S^L)}.
\]

The decoder estimates $\hat{V}$ based on the received $X^{(b)}$ and the channel $P_{X^{(b)}|U^{(b)}}$. Since $P_{X^{(b)}|U^{(b)}}$ is based on the entire sequence, the decoder deterministically selects the most likely bit as

\begin{align}\hat{V}_i =   \underset{V_i}{\textrm{argmax}}\sum_{V_{i+1:L-1}} \mathbb{P}( V_i,V_{i+1:L-1}|X^L,V_{0:i-1}),\nonumber \\\forall i\in \mathcal{L}_{\textrm{dec}}  \label{Decoder}.
\end{align}
Similarly, since $\mathbb{P}(X^{(b)}|U^{(b)}, S^{(b)})$ is in bijection with  $\mathbb{P}(X^{(b)}|V^{(b)}, S^{(b)})$, the decoder can choose \mbox{$\hat{V}^{(b)} = \textrm{argmax}_{V^{(b)}} \mathbb{P}(X^{(b)}|V^{(b)})$}, which corresponds to ML decoding. Note that the decoder is not required to reliably recover $V^{(b)}$ entirely, but only the message bits indexed by the set $\mathcal{M}^*$.

To ensure covertness from an information-theoretic perspective, the encoder one-time pads (OTP) the first $\lceil \mathbb{I}(U^{(b)};X^{(b)}|S^{(b)}) \rceil$ message bits in each block using a shared secret key. The OTP provides both secrecy and the randomness required for resolvability through binning. In the following section, we vary the number of keys to empirically illustrate the resolvability-based covertness predicted by Theorem~\ref{MainTheorem}.

\section{Numerical results}
\subsection{Masked-Diffusion Block Generation}
At each block, we produce the distribution of length-$L$ token sequences under a masked DLM~\cite{nie2025largelanguagediffusionmodels}. Starting from the given context with all $L$ block positions masked, we run the reverse process under a deterministic one-commit-per-step schedule: at each step the model predicts the token marginals at every masked position given the context and previously committed tokens. A candidate's likelihood follows the chain rule as,
\begin{equation}
\log \mathbb{P}(x^{(b,q)}\mid \text{context},\theta)
 = \sum_{t=1}^{L} \log \mathbb{P}\!\left(x^{(b,q)}_{t}\,\middle|\,
 \text{context},\, x^{(b,q)}_{1:t-1},\theta\right).
\end{equation}
 We construct the candidate pool $\mathcal{X}^{(b)} = \{x^{(b,1)},\dots,x^{(b,Q)}\} \subset \mathcal{V}^L$ by retaining top-$Q$ sequence ranked by joint probability, with a sampling temperature setting of $\tau$. 

\subsection{Setup} We use LLaDA-8B~\cite{nie2025largelanguagediffusionmodels}, a masked-diffusion language model, with sampling temperature settings \(\tau = 2, 4, 8\) where higher temperate leads to higher average entropy and \(Q = 20\) retained candidates per block. All experiments use $50$ seeded prompts drawn uniformly from the C4 News dataset. We use \(L = 8\) and \(B = 2\). For the action in Algorithm~\ref{alg:pd_mb}, we apply a masked softmax over the retained support for each block, enabling explicit control of the TV cost while preserving generation quality. Algorithm~\ref{alg:pd_mb} is trained for \(6000\) iterations: the policy is optimized by Adam (learning rate \(10^{-3}\)) and the dual variable \(\beta\) is updated by projected gradient ascent (step size \(0.05\)), under covertness constraint \(\epsilon = 0.1\) and no discount \(\gamma = 1\). Bit-error rates and total variations are averaged over \(1000\) Monte Carlo transmissions per operating point. Thresholds \(T_\delta, T_\epsilon\) are chosen to control the rate.

\subsection{Simulation Results}
Fig.~\ref{fig:Key} considers the setting in which the key rate equals the message rate, and key bits of same length are used to OTP the message bits prior to the polar transform. Due to computational complexity of Algorithm~\ref{alg:pd_mb}, retaining $Q=20$ results in a significantly lower average entropy than top-$20$ sampling under autoregressive models like Llama. Nevertheless, our algorithm achieve below $10\%$ error, $0.375$ bits/token with an average entropy of $0.5$ bits/token.

Fig.~\ref{fig:Key} examines the TV for a range of message rates as the key rate increases. Increasing the number of keys (note that this is not the number of key bits) enhances soft-covering and reduces the average TV, which is minimized near the empirically computed required key rate for covertness under this set up, \(\lceil \mathbb{I}(U^{(b)};X^{(b)}|S^{(b)}) \rceil = 3\). However, for a finite block length \(L=8\), polarization is insufficient to fully realize the average TV constraint \(\epsilon/L\). Consequently, higher message rates force the use of positions where the entropy with respect to the state, \(\mathbb{H}(V_i \mid V^{i-1}, S^L)\), is not sufficiently large. In fact, TV will increase when the number of keys exceeds the threshold of \(8\), as the use of state-constrained positions in \(V\) degrades the effectiveness of soft-covering~\cite{chou2018empirical}. However, this direction is beyond the scope of the paper.

\begin{figure}
    \centering
    \includegraphics[width=1\linewidth]{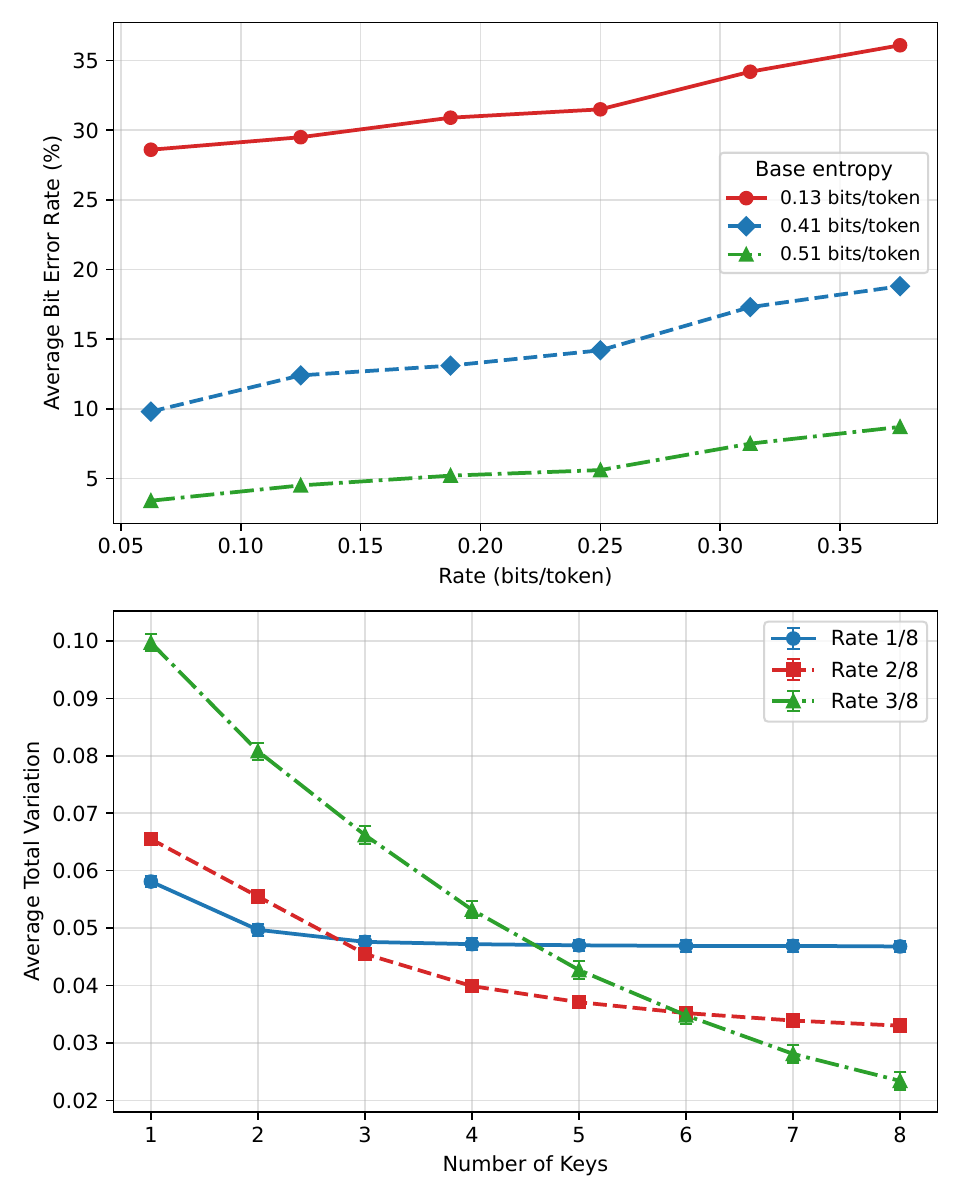}
    \caption{
    Top: Average bit error rate (BER). 
    Bottom: Average total variation versus number of keys used for different embedding rates.
    }
    \label{fig:Key}
\end{figure}

\begin{remark}
\textit{The main limitations of the scheme is the computational difficulty in extending to longer token lengths, due to the exponential growth of the state space across blocks. This renders optimization of the encoding distribution \(P_{S^{(b)}U^{(b)}X^{(b)}}\), as well as Bhattacharyya parameters computation in (\ref{polar set}) expensive as the block length increases~\cite{arikan2009channel}. Nevertheless, the insight offered by the present work is an important step towards designing operational watermarking schemes for DLM that successfully exploit non-causality benefits.}
\end{remark}

\section*{Acknowledgment}{OpenAI's ChatGPT 5.5 was used for improving the language composition of the manuscript. Anthropic's Opus 4.8 was used to assist with coding and debugging simulation scripts. This work was supported in part by the National Science Foundation (NSF) under grant 2148400 as part of the Resilient \& Intelligent NextG Systems (RINGS) program. This work was also supported in part by the National Science Foundation Cybersecurity Innovation for Cyberinfrastructure (CICI) program under Grant No. 2531010. Any opinions, findings, and conclusions or recommendations expressed in this material are those of the authors and do not necessarily reflect the views of the NSF.}

\newpage

\IEEEtriggeratref{22}
\bibliographystyle{IEEEtran}
\bibliography{references_ieee_cleaned.bib}

\newpage
\appendices
\onecolumn

\section{Achievability Proof of Theorem \ref{MainTheorem}}
The proof uses a random coding argument with likelihood encoder and soft-covering in place of covering argument used in standard Gelfand-Pinsker code~\cite{el2011network}. 
\subsection{Code Description}
Consider a random variable $U$ such that the joint distribution $P_{UXS}\triangleq P_SP_{U|S}W_{X|US}$ satisfies $\sum_{u}P_{UXS}(u,x,s)=P_{XS}$.
\begin{enumerate}
    \item Codebook Generation: Construct a codebook $\mathcal{C}_n$ consisting of \(k\in \mathcal{K}\triangleq[1,2^{nR_k}]\) bins, each containing \(m\in \mathcal{M}\triangleq[1,2^{nR}]\) sub-bins of size \(j\in \mathcal{J}\triangleq[1,2^{nR_0}]\). A total of \(2^{n(R_k+R+R_0)}\) codewords are drawn i.i.d from the distribution \(P_U^{\otimes n} = \prod_{t=1}^n P_U(u_t)\).
    \item Encoding: For each state sequence \(s^n\) and key \(K=k\), to send message \(M=m\), the encoder selects codeword \(u^n(m,k,j)\) with likelihood encoder by drawing the index \(J=j\) from 
    \begin{equation}
        P_{J|M,K,S^n}(j|m,k,s^n) = \frac{P_{S|U}^{\otimes n}(s^n|u^n(m,k,j))}{\sum_{j=1}^{|\mathcal{J}|}P_{S|U}^{\otimes n}(s^n|u^n(m,k,j))},
    \end{equation}
    where \(P_{S|U}^{\otimes n}\) is induced by \(P_{S}^{\otimes n}\) and our choice of auxiliary random variable \(U\). The encoder then samples \(x^n\) according to \(W_{X^n|U^n,S^n}(x^n|u^n(m,k,j),s^n)\).
    \item Decoding: The decoder receives \(x^n\) and with access to \(k\), decides message \(m\) if there exists a unique pair \((m,j)\) such that \((x^n, u^n(m,k,j))\in \mathcal{T}_\epsilon^n(XU)\), where \(\mathcal{T}_\epsilon^n(XU)\) is the joint typical set. Otherwise, the decoder declares an error. 
\end{enumerate}

\subsection{Analysis of Covertness}
We first recall known results regarding soft covering~\cite{cuff2013distributed,cuff2015stronger} in the following lemma.
\begin{lemma}[Soft Covering~\cite{cuff2013distributed,cuff2015stronger}]
Let \(\mathcal{C}_n\) be a random collection of sequences \(U^n(i)\), with \(i \in [1,2^{nR}]\), each drawn i.i.d from \(P_U\). Let \(P_{X^n}\) be the induced distribution by selecting \(i\) uniformly at random and transmitting $U^n(i)$ through a memoryless channel with transition probability \(W_{X|U}\). Then for \(R > \mathbb{I}(U;X)\) and \(n\) large enough, we have 
\begin{equation}
    \mathbb{E}_{\mathcal{C}_n}\bigg[\mathbb{V}\bigg(P_{X^n},\prod_{t=1}^n P_X\bigg)\bigg] \leq e^{-n\sigma},
\end{equation}
for some \(\sigma>0\).
 Furthermore, if \(|\mathcal{X}|\) has finite support, then for any \(R > \mathbb{I}(U;X)\) there exists \(\gamma_1, \gamma_2 >0\) such that for \(n\) large enough
\begin{equation}
    \mathbb{P}\bigg(\mathbb{V}\bigg(P_{X^n},\prod_{t=1}^n P_X\bigg) > \frac{1}{\sqrt{2}}e^{-\gamma_1n}\bigg) \leq e^{-e^{\gamma_2 n}}.\label{eq:stronger_covering}
\end{equation}
\end{lemma}

Under our code construction using the likelihood encoder, the induced joint distribution by the codebook is 
\begin{align}\label{induced}
    \widetilde{P}_{MKJS^nU^nX^n}=  \frac{1}{|\mathcal{M}||\mathcal{K}|}P_{S}^{\otimes n}(s^n)P_{J|M,K,S^n}(j|m,k,s^n)\mathbf{1}\{u^n=u^n(m,k,j)\}W_{X|U,S}^{\otimes n}(x^n|u^n,s^n).
\end{align}
We also define the ideal joint distribution induced by a codebook where codeword indices \((m,k,j)\) are chosen uniformly (independent of \(s^n\)) as
\begin{equation}\label{iidinduced}
    \widetilde{Q}_{MKJS^nU^nX^n}=\frac{1}{|\mathcal{M}||\mathcal{K}||\mathcal{J}|}\mathbf{1}\{u^n=u^n(m,k,j)\}P_{S|U}^{\otimes n}(s^n|u^n)W_{X|U,S}^{\otimes n}(x^n|u^n,s^n).
\end{equation}
We first show that (\ref{induced}) and (\ref{iidinduced}) are close in TV. By definition, we have 
\begin{equation}
    \widetilde{Q}_{J|MKS^n}(j|m,k,s^n)= \frac{P_{S|U}^{\otimes n}(s^n|u^n(m,k,j))}{\sum_{\ell}P_{S|U}^{\otimes n}(s^n|u^n(m,k,\ell))}=\widetilde{P}_{J|M,K,S^n}(j|m,k,s^n).
\end{equation}
Therefore, we have from~\cite[Lemma V.1, V.2]{cuff2013distributed} that 
\begin{align}
    \mathbb{V}(\widetilde{P}_{S^nX^n},\widetilde{Q}_{S^nX^n})& \leq \mathbb{V}(\widetilde{P}_{MKJS^nU^nX^n},\widetilde{Q}_{MKJS^nU^nX^n}) \\
    &=\mathbb{V}(\widetilde{P}_{MKS^n},\widetilde{Q}_{MK S^n}) \\
    & =\frac{1}{|\mathcal{M}||\mathcal{K}|}\sum_{m,k}\mathbb{V}\bigg( P_{S}^{\otimes n}(s^n), \frac{1}{|\mathcal{J}|}\sum_{\ell=1}^{|\mathcal{J}|}P_{S|U}^{\otimes n}(s^n|u^n(m,k,\ell))\bigg).
\end{align}
By the soft-covering lemma, we see that when \(R_0 > \mathbb{I}(U;S)\), we have for \(\sigma_1>0\)
\begin{equation}\label{TV1}
    \mathbb{E}_{\mathcal{C}_n}\bigg[\mathbb{V}(\widetilde{P}_{S^nX^n},\widetilde{Q}_{S^nX^n})\bigg] \leq e^{-n\sigma_1}.
\end{equation}

To ensure semantic secrecy, on can derive secrecy from resolvability~\cite{bloch2013strong} and leverage the strong version of the soft-covering lemma~\cite{cuff2015stronger} to show that the induced distribution regardless of the value of the message $m$ can be made similar. Indeed, for
each fixed \(m\), we apply a strong soft-covering in~(\ref{eq:stronger_covering}) over the randomization indices \((k,j)\). Whenever \(R_k+R_0> \mathbb{I}(U;XS)\), we have for \(\gamma_1,\gamma_2 >0\) 
\begin{align} \label{TV2}
    \mathbb{P}_{\mathcal{C}_n}\bigg(\mathbb{V}(\widetilde{Q}_{X^nS^n|M=m},P_{XS}^{\otimes n}) > \frac{1}{\sqrt{2}}e^{-\gamma_1n}\bigg)  
  \leq e^{-e^{\gamma_2 n}}. 
\end{align}
We now apply the union bound over all messages to have 
\begin{equation}\label{doubleexp}
    \mathbb{P}_{\mathcal{C}_n}\bigg(\max_m \mathbb{V}(\widetilde{Q}_{S^nX^n|M=m},P_{XS}^{\otimes n})> \frac{1}{\sqrt{2}}e^{-\gamma_1n}\bigg) \leq |\mathcal{M}|e^{-e^{\gamma_2 n}}. 
\end{equation}
\(|\mathcal{M}|\) is only exponential in \(n\) and therefore th RHS of (\ref{doubleexp}) vanishes in \(n\). 

Finally, taking \(R_0 = \mathbb{I}(U;S)+\frac{\zeta}{2}\) for any \(\zeta > 0\) guarantees that both (\ref{TV1}) and (\ref{TV2}) holds for \(R_k \geq \mathbb{I}(X;U|S)+\frac{\zeta}{2}\). 

\subsection{Analysis of Reliability}
Let \(m,j\) be the true message indices, we define the error event \(\mathcal{E}=\{(\hat{M},\hat{J}) \neq (m,j)\}\). Without loss of generality, we assume \((m,k,j)=(1,1,1)\), then the error probability under joint distribution \(\widetilde{Q}\) is given by 
\begin{equation}
    \widetilde{Q}(\mathcal{E}) \leq \widetilde{Q}((U^n(1,1,1),X^n)\notin \mathcal{T}_\delta^n(P_{UX}))+\bigcup_{(m',j')\neq (1,1)} \widetilde{Q}((U^n(m',1,j'),X^n)\in \mathcal{T}_\delta^n(P_{UX})).
\end{equation}
By asymptotic equipartition property (AEP) and symmetry over \((m,k,j)\)~\cite{thomas2006elements}, if \(R+R_0 < \mathbb{I}(U;X)\), we have \(\mathbb{E}_{\mathcal{C}_n}[\widetilde{Q}(\mathcal{E})]\leq \delta_n\). For any measurable event 
\begin{equation}
    \mathcal{E}\subseteq \mathcal{M}\times \mathcal{K}\times \mathcal{J}\times \mathcal{S}^n\times \mathcal{U}^n \times \mathcal{X}^n,
\end{equation}
from the sup-representation of TV
\begin{equation}
   \mathbb{E}_{\mathcal{C}_n}\bigg[ |\widetilde{P}(\mathcal{E})-\widetilde{Q}(\mathcal{E})|\bigg]\leq  \mathbb{E}_{\mathcal{C}_n}\bigg[\mathbb{V}(\widetilde{P},\widetilde{Q})\bigg]\leq \delta_n,
\end{equation}
and we can conclude \(\mathbb{E}_{\mathcal{C}_n}[\widetilde{P}(\mathcal{E})]\leq \mathbb{E}_{\mathcal{C}_n}[\widetilde{Q}(\mathcal{E})]\leq \delta_n\).

The achievability result finally follows by derandomizing the code using Markov inequality.

\section{Converse Proof of Theorem \ref{MainTheorem}}
Consider a code such that \(P_e \leq \delta_n\) and \(\mathbb{V}\big( \widetilde{P}^m_{X^n S^n}, P_{X^nS^n} \big)\leq \sigma_n\) with \(\delta_n, \sigma_n = o(1)\).
\subsection{Analysis of Covertness}
We derive the converse on the key rate with the following technique.
\begin{align}\label{key bound}
    \log K &= \mathbb{H}(K) = \mathbb{H}(K|M,S^n) \\
    & \geq \mathbb{I}(K;X^n|M,S^n) \\
    & = \mathbb{I}(MK;X^n|S^n) - \mathbb{I}(M;X^n|S^n)\\
    & \geq \mathbb{I}(MK;X^n|S^n) - \mathbb{I}(M;X^n,S^n) \\
    & \geq \mathbb{I}(MK;X^n|S^n) -o(n),
\end{align}
where last line follows from the fact that \(\mathbb{V}\big( \widetilde{P}^m_{X^n S^n}, P_{X^nS^n} \big) \leq \sigma_n\), \(\sigma_n=o(1)\) and \cite[lemma 2.7]{csiszar1982information}. Continuing (\ref{key bound})
\begin{align}\label{converseresolvability}
    \log K + o(n) 
    &\ge \mathbb{I}(MK;X^n|S^n) \\
    &= \sum_{i=1}^n \mathbb{I}(MK;X_i|X^{i-1},S^n) \displaybreak[0]\\
    &= \sum_{i=1}^n \bigg[\mathbb{H}(X_i|X^{i-1},S^n) - \mathbb{H}(X_i|MK,X^{i-1},S^n)\bigg] \displaybreak[0]\\
    &= \sum_{i=1}^n \bigg[\mathbb{H}(X_i|S_i) - \mathbb{I}(X_i;X^{i-1},S^{i-1},S_{i+1}^n|S_i)- \mathbb{H}(X_i|MK,X^{i-1},S^n)\bigg]\displaybreak[0]\\
    &\overset{(a)}{\geq}\sum_{i=1}^n \bigg[\mathbb{H}(X_i|S_i) - \mathbb{H}(X_i|MK,X^{i-1},S_{i}^n)\bigg] 
    - \sum_{i=1}^n \mathbb{I}(X_i;X^{i-1},S^{i-1},S_{i+1}^n|S_i) \displaybreak[0]\\
    &\overset{(b)}{=} \sum_{i=1}^n \mathbb{I}(U_i;X_i|S_i) 
    - \sum_{i=1}^n \mathbb{I}(X_i;X^{i-1},S^{i-1},S_{i+1}^n|S_i) \displaybreak[0]\\
    & \overset{(c)}{=} \sum_{i=1}^n \mathbb{I}(U_i;X_i|S_i) 
    - \sum_{i=1}^n \mathbb{I}(X_i;A_i|S_i)
\end{align}
where \((a)\) follows because conditioning does not increase entropy, \((b)\) is by identifying \(U_i \triangleq (M,K,S_{i+1}^n,X^{i-1})\) and \((c)\) is by setting \(A_i \triangleq (X^{i-1},S^{i-1},S_{i+1}^n)\). Under product law \(P_{XS}^{\otimes n}\), we have \(X_i \bot A_i|S_i\), we observe that 
\begin{equation}\label{36}
    \mathbb{I}(X_i;A_i|S_i) \leq |\mathbb{H}(X_i|S_i)-\mathbb{H}_{P_{XS}}(X_i|S_i)|+ |\mathbb{H}(X_i;A_i|S_i)-\mathbb{H}_{P_{XS}}(X_i;A_i|S_i)|.
\end{equation}
From the premise of the converse \(\mathbb{V}(\widetilde{P}_{X^nS^n} ,P_{XS}^{\otimes n}) \leq \sigma_n\) and~\cite[Lemma V.1,V.2]{cuff2013distributed}, we have 
\begin{equation}
    \mathbb{V}(\widetilde{P}_{X_iS_i} ,P_{X_i,S_i})\leq \sigma_n,  \mathbb{V}(\widetilde{P}_{X_iA_iS_i} ,P_{X_iA_iS_i}) \leq \sigma_n.
\end{equation}
Notice that if \(\mathbb{V}(P,Q) \leq \sigma_n\) on an alphabet of \(|\mathcal{A}|\), then \(|\mathbb{H}(P)-\mathbb{H}(Q)|\leq \sigma_n \log |\mathcal{A}|+h_2(\sigma_n)\). Applying the entropy bound twice \cite[lemma 2.7]{csiszar1982information} 
\begin{equation}
     |\mathbb{H}(X_i|S_i)-\mathbb{H}_{P_{XS}}(X_i|S_i)| \leq 2\sigma_n \log |\mathcal{X}| + 2h_2(\sigma_n).
\end{equation}
By the same treatment of the second term on the RHS of (\ref{36}) and the fact that \(\sigma_n = o(1)\), it can be shown that \(\sum_{i=1}^n \mathbb{I}(X_i;A_i|S_i) =o(n)\) \footnote{Note that a scaling of \(o(1)\) in $\sum_{i=1}^n \mathbb{I}(X_i;A_i|S_i)$ and (\ref{key bound}) is possible with a stronger soft-covering result for achievable definition with \(\sigma_n = o(1/n)\).}. 
Letting $Q$ be a random variable uniformly distributed over $[1;n]$ and setting $U\triangleq U_Q$, $X\triangleq X_Q$, $S\triangleq (S_Q,Q)$, we conclude that for any $\epsilon>0$ and $n$ large enough
\begin{align}
    \frac{1}{n}\log K 
&\ge \mathbb{I}(U;X|S) - \epsilon
\end{align}
where $S,X,U$ are such that $S,X$ have joint distribution $P_{X|S}P_S$.

\subsection{Analysis of Reliability}
We now turn to the converse for rate bound, which follows standard Gelfand-Pinsker converse technique. Starting from Fano's inequality 
\begin{align}\label{fano}
    \log M = &\mathbb{H}(M|K)\\
       = & \mathbb{H}(M|K,X^n)+\mathbb{I}(M;X^n|K) \\
     \leq & \mathbb{I}(M;X^n|K) +n\delta_n.
\end{align}
Then, continuing (\ref{fano})
\begin{align}
    \log M -n \delta_n & \leq \mathbb{I}(M;X^n|K) \\
    & \leq \sum_{i=1}^n \mathbb{I}(M,X^{i-1};X_i|K) \\
     & = \sum_{i=1}^n \mathbb{I}(M,X^{i-1},S_{i+1}^n;X_i|K)- \sum_{i=1}^n\mathbb{I}(X_i;S_{i+1}^n|K,M,X^{i-1}) \\
     & \overset{(a)}{\leq} \sum_{i=1}^n \mathbb{I}(M,X^{i-1},S_{i+1}^n,K;X_i)- \sum_{i=1}^n\mathbb{I}(X^{i-1};S_i|S_{i+1}^n,K,M)\\
      & \overset{(b)}{=} \sum_{i=1}^n \mathbb{I}(U_i;X_i)- \sum_{i=1}^n\mathbb{I}(X^{i-1},S_{i+1}^n,K,M;S_i)\\
      & = \sum_{i=1}^n \mathbb{I}(U_i;X_i)- \sum_{i=1}^n\mathbb{I}(U_i;S_i)
\end{align}
where (a) is by Csiszar's sum identity and (b) is by \(U_i \triangleq (M,K,S_{i+1}^n,X^{i-1})\) and \(\mathbb{I}(M,K, S_{i+1}^n;S_i) = 0\). Therefore, we can conclude 
\begin{equation}
    \frac{1}{n}\log M \leq  \mathbb{I}(U;X)-\mathbb{I}(U;S) +\delta_n.
\end{equation}

\section{Proof of Corollary \ref{corollary 5.2}}
\paragraph{Achievability} Partition the tokens into two sets \(\mathcal{A}, \mathcal{B}\) with \(|\mathcal{A}|=\lfloor \frac{V}{2}\rfloor\), \(|\mathcal{B}|=\lceil \frac{V}{2}\rceil\). We define the auxiliary 
\begin{equation}
    U = \begin{cases}
        0, & X\in\mathcal{A}\\
        1, & X\in \mathcal{B}.
    \end{cases}
\end{equation}
Since \(U\) is deterministic function of \(X\), we have \(\mathbb{H}(U|X)=0\). Moreover, we have for any \(\textnormal{supp}(S)=\{i,j\}\)
\begin{equation}
    \mathbb{H}(U|S) = \begin{cases}
        1, &  i \in \mathcal{A}, j\in \mathcal{B} \,\textnormal{or}\,j\in \mathcal{A},i\in \mathcal{B}\\
        0, &  i,j\in \mathcal{A} \,\textnormal{or}\, i,j\in \mathcal{B}.
    \end{cases}
\end{equation}
We then conclude with the identity 
\begin{equation}
    \mathbb{I}(U;X)-\mathbb{I}(U;S) = \mathbb{H}(U|S)-\mathbb{H}(U|X) = \frac{|\mathcal{A}||\mathcal{B}|}{\binom{V}{2}} = \frac{\lfloor V^2/4 \rfloor}{\binom{V}{2}}.
\end{equation}

\paragraph{Converse} The converse part of the proof follows by first showing any stochastic mapping \(X \rightarrow U\) cannot increase the objective. Define 
\begin{equation}
    q_i=\mathbb{P}(U=1|X=i), i=1,\dots, V, 
\end{equation}
the objective becomes 
\begin{align}
    & \mathbb{H}(U|S)-\mathbb{H}(U|X) \\
    ={}& \frac{1}{\binom{V}{2}}\sum_{1\leq i<j\leq V}
    h_2\bigg(\frac{q_i+q_j}{2}\bigg)
    -\frac{1}{V}\sum_{i=1}^V h_2(q_i) \\
    ={}& \frac{1}{\binom{V}{2}}\sum_{1\leq i<j\leq V}
    \bigg[
    h_2\bigg(\frac{q_i+q_j}{2}\bigg)
    -\frac{h_2(q_i)+h_2(q_j)}{2}
    \bigg].
\end{align}
The Jensen-Shannon Divergence \(J(a,b)=h_2(\frac{a+b}{2})-\frac{h_2(a)+h_2(b)}{2}\)~\cite{polyanskiy2025information} is maximized for \(a,b\in \{0,1\}\). Therefore we can conclude 
\begin{equation}
    \mathbb{I}(U;X)-\mathbb{I}(U;S) \leq \mathbb{H}(U|S)\leq \frac{\lfloor V^2/4 \rfloor}{\binom{V}{2}}.
\end{equation}
Finally, we have the key rate \(\mathbb{I}(U;X|S)=\mathbb{I}(U;X)-\mathbb{I}(U;S)\) since \(\mathbb{I}(U;S|X)=0\).


\end{document}